\documentclass[letterpaper,10 pt,conference]{ieeeconf}
\IEEEoverridecommandlockouts
\usepackage{amsmath}
\usepackage{amssymb}
\usepackage{mathptmx}
\usepackage{graphicx}
\usepackage{subfig}
\usepackage[hidelinks]{hyperref}

\newtheorem{conjecture}{Conjecture}

\newtheorem{remark}{Remark}
\newtheorem{theorem}{Theorem}

\title{\LARGE \bf Distributed Area Coverage Control with Imprecise Robot Localization}
\author{Sotiris Papatheodorou, Yiannis Stergiopoulos, and Anthony Tzes
\thanks{The authors are with the Electrical \& Computer Engineering Department, University of Patras, Greece}
\thanks{Corresponding author's e-mail: tzes@ece.upatras.gr}
\thanks{This work has received funding from the European Union Horizon 2020 Research and Innovation Programme under the Grant Agreement No. 644128, AEROWORKS}
}
\begin{document}

\maketitle
\thispagestyle{empty}
\pagestyle{empty}

\begin{abstract}
	This article examines the problem of area coverage for a network of mobile robots with imprecise agent localization.
	Each robot has uniform radial sensing ability, governed by first order kinodynamics. The convex-space is partitioned based on the Guaranteed Voronoi (GV) principle and each robot's area of responsibility corresponds to its GV-cell, bounded by hyperbolic arcs.
	The proposed control law is distributed, demands the positioning information about its GV-Delaunay neighbors and has an inherent collision avoidance property.
\end{abstract}


\section{Introduction}

The problem of area coverage by a swarm of mobile robots is one of great interest during the last years.
Autonomous robotic agents with on board sensors disperse in areas of interest in order to detect important events.
The outmost goal is the development of distributed algorithms which have a low computational cost, and can easily adapt to changes in the robot swarm topology, while rely only on information from neighboring agents (and not demand global network information) in order to lead the agents to optimal positions.
In these cases the agents are equipped with radio transceivers in additions to their sensors, in order to exchange crucial information with neighboring nodes \cite{Stergiopoulos_IEEETAC15,Kantaros_Automatica15,Stergiopoulos_Hindawi12,Li_IETC09}.

Several distributed coordination algorithms have been developed, both for dynamic \cite{Zhai_Automatica2013} and static coverage problems \cite{Song_Automatica2013}.
The agents' sensors in the majority of the literature are assumed to have uniform radial sensing patterns \cite{Cortes_ESAIMCOCV05,Stergiopoulos_ACC09}.
There have been extensions amending for disc footprints of unequal radii \cite{Pimenta_CDC08,Stergiopoulos_IETCTA10} as well as arbitrary sensing ones \cite{Stergiopoulos_ICRA14}.
Although most commonly the region of interest is convex, non convex regions and/or domains containing obstacles have been studied \cite{Stergiopoulos_IEEETAC15}.

Control laws for area coverage \cite{Cortes_TRA04,Martinez_CSM07} can be enhanced to account for the positioning uncertainty of the agents' nodes based on the Guaranteed Voronoi tessellation \cite{Evans_CCCG2008} and a gradient ascent control law.

The article is organized as follows. The main assumptions, mathematical preliminaries and introduction to Guaranteed Voronoi partitioning are presented in Section \ref{section:problem_statement}. The Guaranteed Voronoi diagram when the nodes are assumed as disks is examined in detail in Section \ref{section:GV_disks} and a distributed control law is presented in Section \ref{section:distributed_control_law} that increased the network's coverage performance through its evolution. Simulation studies indicating the efficiency of the proposed scheme are presented in Section \ref{section:simulation_studies}, followed by concluding remarks.

\section{Problem Statement}
\label{section:problem_statement}

\subsection{Main assumptions - Preliminaries}
Let $\Omega \subset \mathbb{R}^2$ be a compact convex region under surveillance.
Assume a swarm of $n$ dimensionless mobile robots.
The robots' location is not known precisely but each robot is guaranteed to be within a positioning uncertainty circle $C_{i}^u$ with center $x_i$ and radius $r^u$
\begin{equation}
	C_{i}^u(x_i,r_u) = \left\{ x \in \Omega: \parallel x-x_i\parallel \leq r^u \right\},~i \in I_n,
	\label{dimension}
\end{equation}
where $\parallel \cdot \parallel$ corresponds to the Euclidean metric and $I_n = \left\{ 1, \dots ,n \right\}$.
The radius $r^u$ is common among all robots, that is all robots have a common positioning error.

The individual robot's motion is considered to be controlled via its velocities as
\begin{equation}
	\dot{x}_i=u_i,~~u_i \in \mathbb{R}^2,~ x_i \in \Omega,~~i \in I_{n},
	\label{kinematics}
\end{equation}
where $u_i$ is the corresponding control input for each robot~(node).

As far as the sensing performance of the nodes is concerned, all members are assumed to have identical range--limited uniform radial sensing footprints defined as
\begin{equation}
	C_{i}^{s}(x_i,r^s) = \left\{ x \in \Omega: \parallel x-x_i\parallel \leq r^s \right\},~i \in I_{n},
	\label{sensing}
\end{equation}
where $r^s$ is the common sensing radius of the members of the network.

We define the guaranteed sensing region of a node with positioning uncertainty $C_{i}^u$ and sensing footprint $C_{i}^{s}$ as
\begin{equation*}
	C_{i}^{gs}(C_{i}^u,C_{i}^{s}) = \left\{ \bigcap_{x_i} C_{i}^{s}(x_i,r^s), ~\forall x_i \in C_{i}^u \right\}, i \in I_{n}
\end{equation*}
Since both $C_{i}^u$ and $C_{i}^{s}$ are circular disks, it can be shown that
\begin{equation}
	C_{i}^{gs}(x_i,r^u,r^s) = \left\{ x \in \Omega: \parallel x-x_i\parallel \leq r^s-r^u \right\},~i \in I_{n},
	\label{guaranteed_sensing}
\end{equation}
where $x_i$ is the center of $C_{i}^u$.

\subsection{Voronoi Diagram}

When the nodes' positions are known precisely, that is $(r^u = 0, i \in I_n)$, the space can be assigned among the nodes via the Voronoi tessellation \cite{Aurenhammer_ELSEVIER99}.
The responsibility region (Voronoi cell) for an agent is defined as the part of the space that is closer to that agent than any other agent of the team (in the Euclidean sense) as
\begin{equation*}
	V_i=\left\{x\in\Omega\colon\left\|x-x_i\right\|
	\leq\left\|x-x_j\right\|,~~\forall
	j\in I_n,~j\neq i\right\},~~i\in I_n.
\end{equation*}
The Voronoi diagram for a set of 6 nodes is shown graphically in the left part of Figure \ref{fig:V_GV_comparison}.

\begin{figure}[htb]
	\centering
	\includegraphics[width=0.23\textwidth]{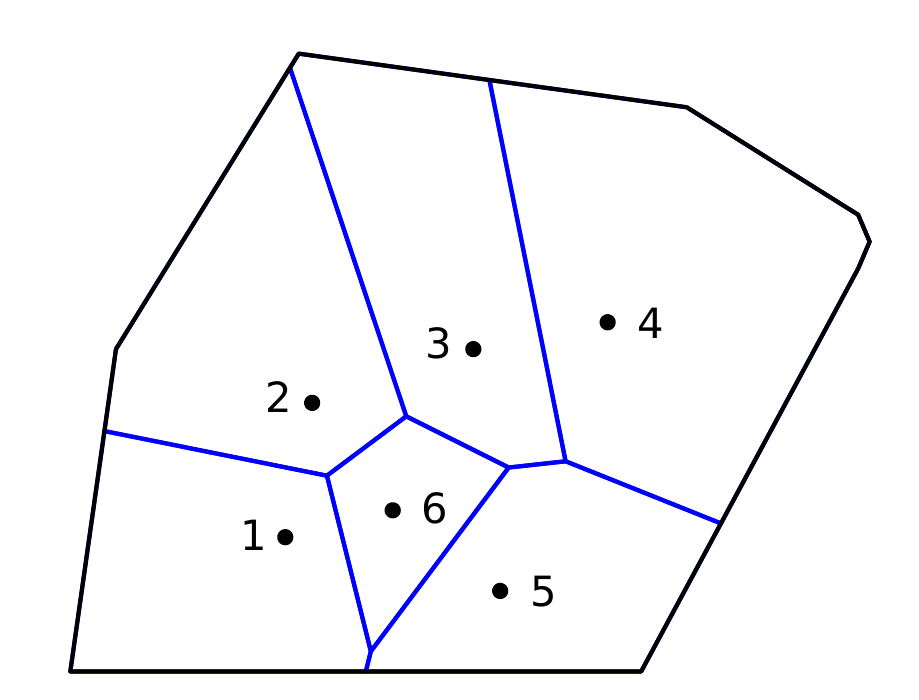}\hspace{0.01cm}
	\includegraphics[width=0.23\textwidth]{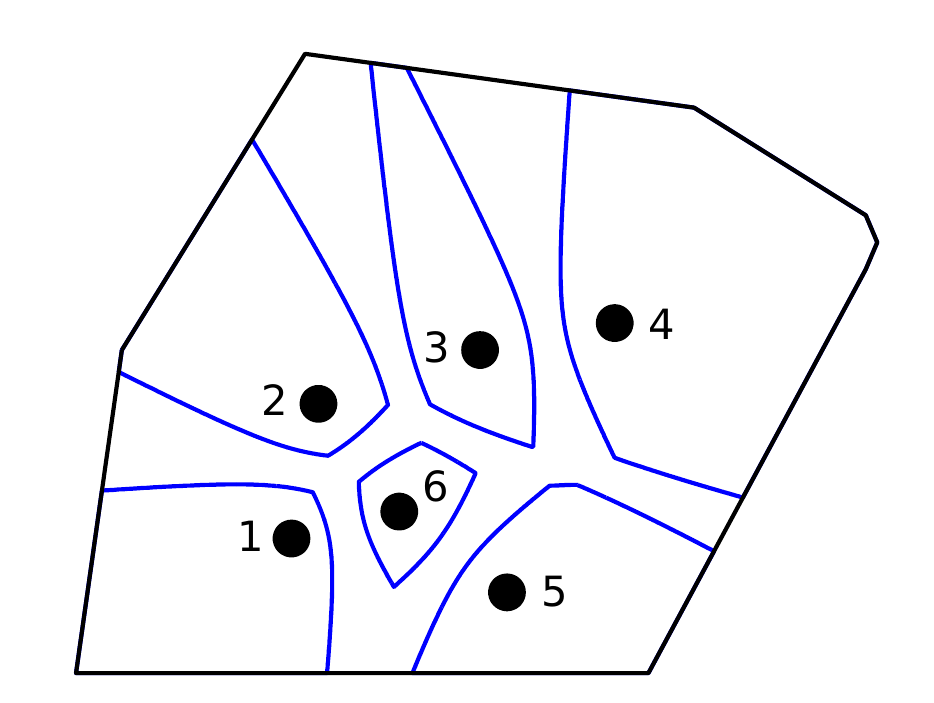}
	\caption{Voronoi diagram (left) for 6 dimensionless nodes, with its equivalent Guaranteed Voronoi diagram (right) in the case of disks (centered on those points).}
	\label{fig:V_GV_comparison}
\end{figure}

The main tessellation properties of the Voronoi diagram that are of particular interest are summarized as: a) $\bigcup_{i\in I_n}V_i=\Omega$, and b) $\mathrm{Int}\left(V_i\right)\cap\mathrm{Int}\left(V_j\right)=\varnothing,~\forall
i,j\in I_n,~i\neq j$, where
$\mathrm{Int}\left(\cdot\right)$ is the interior of the set--argument.
Equivalently to the Voronoi diagram, one can define the Delaunay neighbors of a node $i$, denoted as $N_i$, as
\begin{equation}
 N_i = \left\{j\in I_n,~j\neq i \colon V_i\cap V_j \neq \emptyset\right\},~i\in I_n,
\label{delaunay}
\end{equation}
which are the nodes of the network whose Voronoi cells share an edge with that of node $i$.

\subsection{Guaranteed Voronoi Diagram}
\label{guaranteed_voronoi_cell}

The Guaranteed Voronoi diagram (GV) is defined for a set of uncertain regions $D = \left\{ D_1, \ldots, D_n, ~~ D_i \subset \mathbb{R}^2 \right\}$, i.e. each uncertain region contains all the possible locations of a point $x_i \in \mathbb{R}^2$.
A cell $V_i^g$ is assigned at each uncertain region so that it contains the points of the plane that are closer to $x_i$ for any possible configuration of the points $x_1,\ldots,x_n ~i \in I_n$.
Thus the set of points in the plane that are at least as close to region $D_i$ as region $D_j$, in a guaranteed sense, is
\small
\begin{equation}
	H_{ij}^g = \left\{ x\in\Omega\colon \left\|x-x_i\right\|
	\leq\left\|x-x_j\right\|, ~\forall x_i \in D_i, ~\forall x_j \in D_j
	\right\}.
\end{equation}
\normalsize
The cell of the set $D_i$ is then defined as the intersection of all $H_{ij}^g$ as
\begin{eqnarray}
	\nonumber
	V_i^g &=& \bigcap_{j \neq i} H_{ij}^g\\
	\nonumber
	&=& \left\{x\in\Omega \colon \right. \max \left\|x-x_i\right\|  \leq \min \left\|x-x_j\right\| \\
	&&\forall j\in  I_n,\left. ~j\neq i, ~x_i \in D_{i},~x_j \in D_{j}\right\}
	\label{GV_cell}
\end{eqnarray}

Similarly to the Delaunay neighbors $N_i$, we can define the Guaranteed Delaunay Neighbors $N_i^g$ as
\begin{equation}
N_i^g = \left\{ j \in I_n, j \neq i \colon \partial H_{ij} \cap \partial V_i^g \neq \emptyset \right\}
\label{guaranteed_delaunay}
\end{equation}
which are the nodes of the network that influence $\partial V_i^g$.

\begin{remark}
	\label{neutral_region}
	The GV--diagram has the following properties:
	a) $\bigcup_{i\in I_n}V_i^g\subseteq\Omega$, and
	b) $V_i^g \subseteq V_i ~~\forall i \in I_{n}$.
	Since the GV is not a complete tessellation of the space $\Omega$, let the neutral zone $\mathcal{O} \stackrel{\triangle}{=} \Omega \setminus \left( \cup_{i \in I_n} V_i^g \right)$ correspond to the set of points of the space under consideration not assigned at any node of the network.
\end{remark}

\section{Guaranteed Voronoi diagram of disks}
\label{section:GV_disks}

When the uncertain regions $D_i$ are disks, that is $D_i=C_{i}^{gs}(x_i,r^u,r^s)$, it is shown \cite{Evans_CCCG2008} that $\partial H_{ij}^g$ and $\partial H_{ji}^g$ are the two branches of a hyperbola with foci at $x_i$ and $x_j$ and eccentricity $e = \frac{ \left\|x_i-x_j\right\|}{ r_i^d+r_j^d }=\frac{ \left\|x_i-x_j\right\|}{2r^d}$.
Graphically, the Guaranteed Voronoi diagram for a set of 6 disks is shown in the right part of Figure \ref{fig:V_GV_comparison}, compared to the one for the case of non--dimensional nodes (Voronoi diagram, at left part).

\begin{remark}
	In the particular case of uncertain disks, $\partial H_{ij}^g$ and $\partial H_{ji}^g$ are the two branches of the same hyperbola (even when $r_i^d \neq r_j^d$) and are thus symmetric with respect to the perpendicular bisector of $x_i x_j$.
\end{remark}

\begin{remark}
	In the case of uncertain disks, it has been proven \cite{Evans_CCCG2008} that $N_i^g \subseteq N_i$ with $N_i^g = N_i$ being always true when $\Omega = \mathbb{R}^2$.
	It is clear from Figure \ref{fig:V_GV_comparison} that while nodes 1 and 5 are Delaunay neighbors in the Voronoi diagram, they aren't Guaranteed Delaunay neighbors in the GV diagram when it is calculated for a region $\Omega$ of finite size.
	Since there are already simple $O(n \log_2 n)$ algorithms for the computation of $N_i$, whereas the computation of $N_i^g$ is more complex, the Delaunay neighbors can be used for the creation of the Guaranteed Voronoi diagram.
	\label{rem:delaunay}
\end{remark}

The dependence of the GV cells of two disks on the distance of their centers $d_{ij} = d(x_i,x_j)$ is as follows.
When the disks $C_i, C_j$, overlap both of the cells $V_i^g, V_j^g$ are empty as seen in Fig \ref{fig:cells_distance} (a).
When the disks $C_i, C_j$ are outside tangent, the resulting cells are rays starting from the centers of the disks $x_i, x_j$ and extending along the direction of $x_ix_j$ as seen in Fig \ref{fig:cells_distance} (b).
When the disks are disjoint, the GV cells are bounded by the two branches of a hyperbola.
As $d_{ij}$ increases further, the eccentricity of the hyperbola increases and the distance of the disk centers from the hyperbola vertices (the points of the hyperbola closest to it's center) increases as seen in Fig \ref{fig:cells_distance} (c), (d).
This results in an increase of the area of $V_i^g$ covered by $C_{i}^{gs}$.
The distance between the hyperbola's vertices remains constant at $2r^d$ as $d_{ij}$ increases.

The GV cells of two disks also depend on the sum of their radii $r_i^u + r_j^u$ as seen in Figure \ref{fig:cells_dimensions}.
As $r_i^u + r_j^u$ decreases, the eccentricity of the hyperbola increases and the result on the cells is the same as if the distance of the disks' centers increased, as explained previously.
When $r_i^u + r_j^u = 0$, the GV cells are the classic Voronoi cells.

\begin{figure}[htb]
	\centering
	\subfloat[]{ \includegraphics[width=0.23\textwidth]{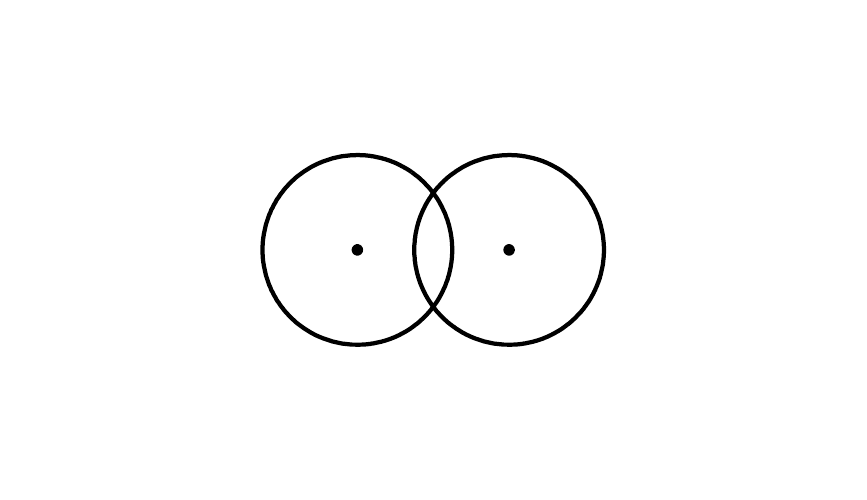} }
	\subfloat[]{ \includegraphics[width=0.23\textwidth]{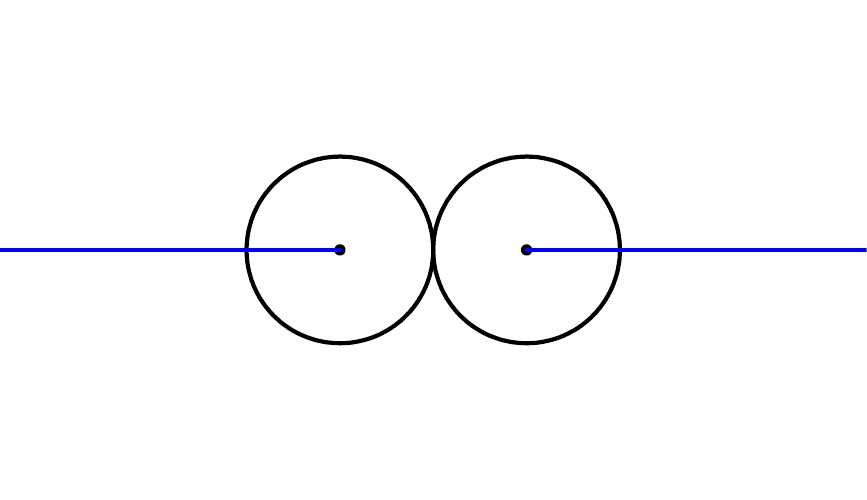} }\\
	\subfloat[]{ \includegraphics[width=0.23\textwidth]{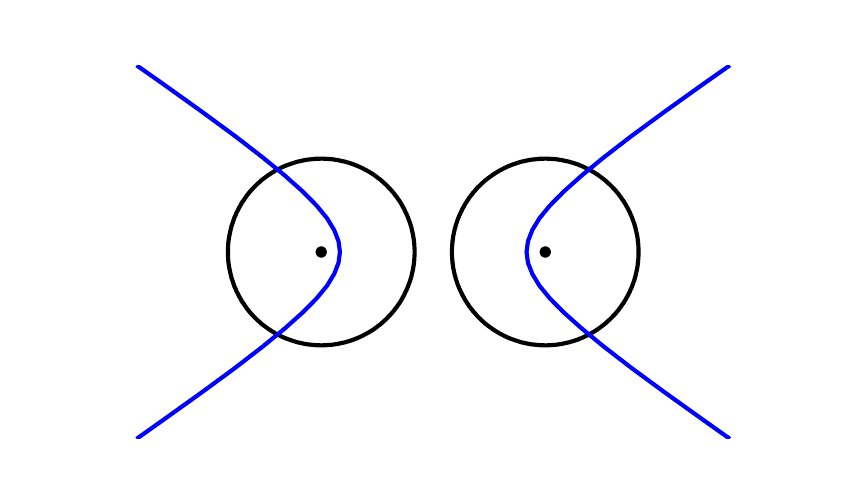} }
	\subfloat[]{ \includegraphics[width=0.23\textwidth]{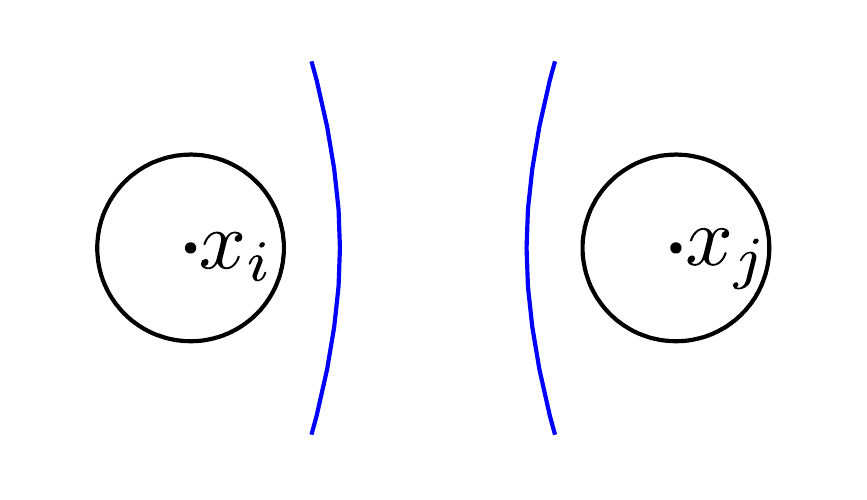} }
	\caption{Dependence of the GV cells on $\left\|x_i-x_j\right\|$.}
	\label{fig:cells_distance}
\end{figure}

\begin{figure}[htb]
	\centering
	\includegraphics[width=0.155\textwidth]{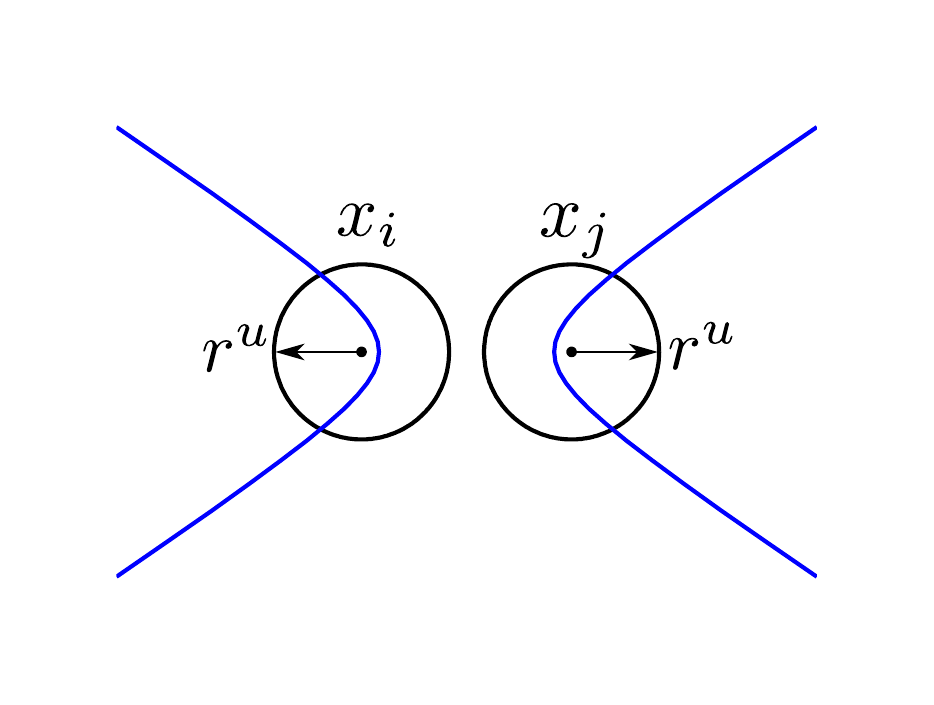}\hspace{0.01cm}
	\includegraphics[width=0.155\textwidth]{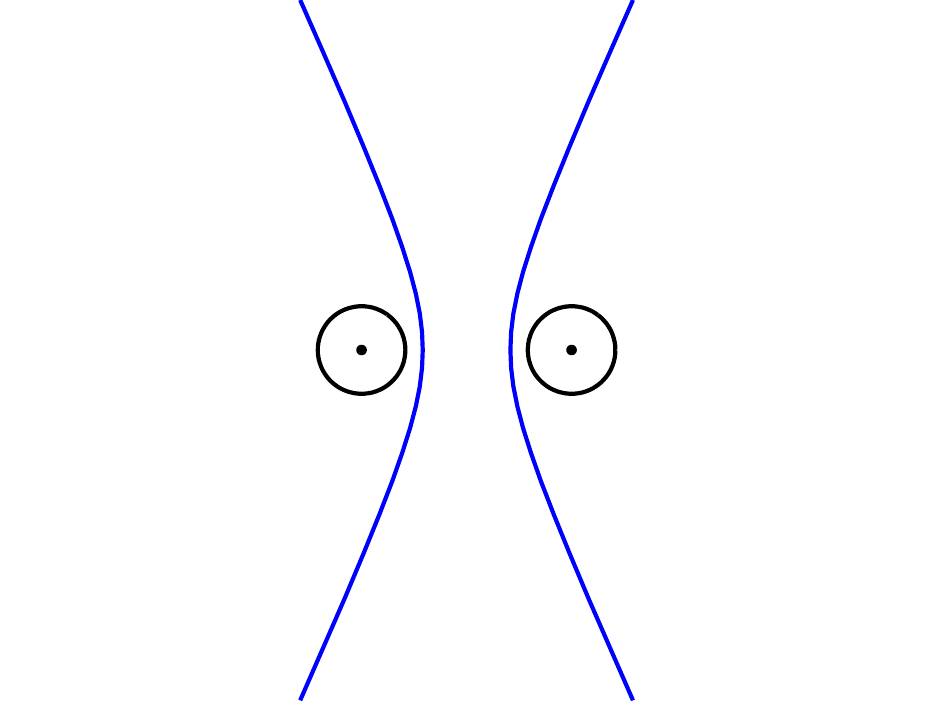}\hspace{0.01cm}
	\includegraphics[width=0.155\textwidth]{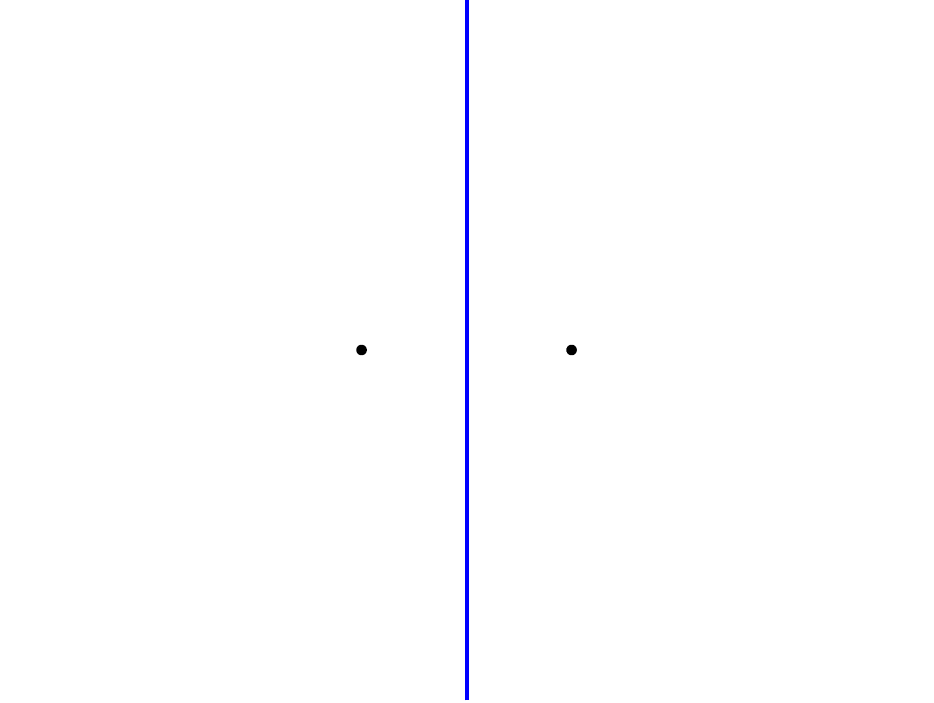}
	\caption{Dependence of the GV cells on $r_i^u + r_j^u$.}
	\label{fig:cells_dimensions}
\end{figure}

With this assumption the area of the total region of $\Omega$ surveyed by the network is
\begin{equation}
	\mathcal{H}=\mathcal{A}\left(\Omega\cap\bigcup_{i\in I_n}C_{i}^{gs}\right)=\sum_{i\in I_n}\mathcal{A}\left(V_i^{gs}\right),
	\label{coverage_criterion}
\end{equation}
where $\mathcal{A}(\cdot)$ is the area function of the set--argument, while the sets $V_i^{gs}$ defined as
\begin{equation}
	V_i^{gs}=V_i^g\cap C_{i}^{gs},~~i\in I_n,
\label{guaranteed sensed cells}
\end{equation}
are the parts of the guaranteed Voronoi cells that are sensed by the nodes themselves.

\section{Spatially Distributed Coordination Algorithm}
\label{section:distributed_control_law}

Based on the kinematic model of the nodes (\ref{kinematics}), their guaranteed sensing regions (\ref{guaranteed_sensing}), and the coverage--oriented criterion under consideration (\ref{coverage_criterion}), a distributed control action is designed in order to lead the mobile team towards a coverage optimal configuration.
The control development takes into account the way that the space is partitioned and assigned among the nodes, which relies on the nodes' positioning uncertainty.

\subsection{Coverage control law development}
Let the ``guaranteed coverage objective" be written in integral form as
\small
\begin{equation}
	\mathcal{H} =
		\sum\limits_{i\in I_n}\int_{V_i^{gs}}\phi\left(x\right)\mathrm{d}x=
		\sum\limits_{i\in I_n}\int_{V_i^{g} \cap C_i^{gs}}\phi\left(x\right)\mathrm{d}x=
		\sum\limits_{i\in I_n}\mathcal{H}_{i},
	\label{Hgc decomposition}
\end{equation}
\normalsize
expressing the total area covered by the nodes in their cells of responsibility, assigned via Guaranteed Voronoi partitioning (\ref{GV_cell}).
The function $\phi\colon \mathbb{R}^2\rightarrow \mathbb{R}_+$ is related to the a--priori knowledge of importance of a point $x\in\Omega$ indicating the probability of an event to take place at $x$ in a coverage scenario.

\begin{theorem}
	Considering a mobile sensor network consisting of nodes with uniform range--limited radial performance as in (\ref{sensing}), governed by the individual robot's kinodynamics described in (\ref{kinematics}), and the GV--partitioning of $\Omega$ defined in (\ref{guaranteed_voronoi_cell}), the coordination scheme
	\begin{eqnarray}
		\nonumber
		u_i &=& \alpha_i \int_{\partial V_i^{gs} \cap\partial C_i^{gs}}n_i\,\phi\,\mathrm{d}x + \\
		&& \hspace{-2cm} \alpha_i  \sum_{j\in N_i^g}\left[\int_{\partial V_i^{gs} \cap \partial H_{ij}} \tilde{\upsilon}_i^i\,\tilde{n}_i\,\phi\,\mathrm{d}x + \int_{\partial V_j^{gs} \cap \partial H_{ji}} \tilde{\upsilon}_j^i\,\tilde{n}_j\,\phi\,\mathrm{d}x\right]
		\label{control_law}
	\end{eqnarray}
	where $n_i$ is the outward unit normal on $\partial V_i^{gs}$ and $\alpha_i$ a positive constant, maximizes the performance criterion (\ref{Hgc decomposition}) along the nodes' trajectories in a monotonic manner, leading in a locally area--optimal configuration of the network.
	\label{ch4:thm:proposed law}
\end{theorem}

\begin{proof}
	Considering (\ref{Hgc decomposition}) and taking into account that the sets $V_i^{gs},~i\in I_n$ are mutually disjoint, the partial derivative of $\mathcal{H}$ with respect to $x_i$ is written as
	\[\frac{\partial\mathcal{H}}{\partial x_i}=
	\frac{\partial}{\partial x_i}\int_{V_i^{gs}}\phi  dx+ \sum_{j \in N_i^g}\frac{\partial}{\partial x_i}\int_{V_j^{gs}}\phi dx.\]
	Considering the second summation term, infinitesimal motion of $x_i$ may only affect $\partial V_j^{gs}$ at $\partial V_j^{gs} \cap \partial \mathcal{O}$, where $\mathcal{O}$ is the neutral area defined in Remark \ref{neutral_region}, since both hyperbola branches are affected by alteration of one of the foci. In addition, only the Guaranteed Delaunay neighbors (\ref{guaranteed_delaunay}) of $i$ are considered in the summation, as a major property of GV--partitioning. Therefore, the former expression can be written via the generalized Leibniz integral rule \cite{Flanders_AMM73} (by converting surface integrals to line ones) as
	\[\frac{\partial \mathcal{H}}{\partial x_i}=
	\int_{\partial V_i^{gs}}\tilde{\upsilon}_i^i\,n_i\,\phi dx
	+
	\sum_{j \in N_i^g}\int_{\partial V_j^{gs} \cap \partial \mathcal{O}} \tilde{\upsilon}_j^i\,n_j\,\phi dx,\]
	where $\tilde{\upsilon}_i^i,\tilde{\upsilon}_j^i$ stand for the transpose Jacobian matrices with respect to $x_i$ of the points $x\in \partial V_i^{gs}$, $x\in \partial V_j^{gs}$, respectively, i.e.
	\begin{equation}
		\tilde{\upsilon}_j^i\left(x\right) \stackrel{\triangle}{=} { \frac{\partial x}{\partial x_i} }^T,~~x\in \partial \tilde{V}_j^{gs},~i,j\in I_n.
		\label{ch3:eq:upsilon tilde}
	\end{equation}
	The boundary $\partial \tilde{V}_i^{gs}$ can be decomposed in disjoint sets as
	\begin{equation}
	\hspace{-0.05cm}\partial V_i^{gs}=
	\left\{\partial V_i^{gs}\cap\partial\Omega\right\}
	\cup
	\left\{\partial V_i^{gs}\cap\partial C_i^{gs}\right\}
	\cup
	\left\{\partial V_i^{gs} \cap \partial \mathcal{O} \right\}.
	\end{equation}
	These sets represent the parts of $\partial V_i^{gs}$ that lie on the boundary of $\Omega$, the boundary of the node's sensing region, and the boundary of the unassigned region of $\Omega$, respectively. This decomposition is shown in Figure \ref{fig:boundary_decomposition} with $\partial V_i^{gs}\cap\partial\Omega$ in green, $\partial V_i^{gs}\cap\partial C_i^{gs}$ in solid red and $\partial V_i^{gs} \cap \partial \mathcal{O}$ in solid blue.

	\begin{figure}[htb]
		\centering
		\includegraphics[width=0.4\textwidth]{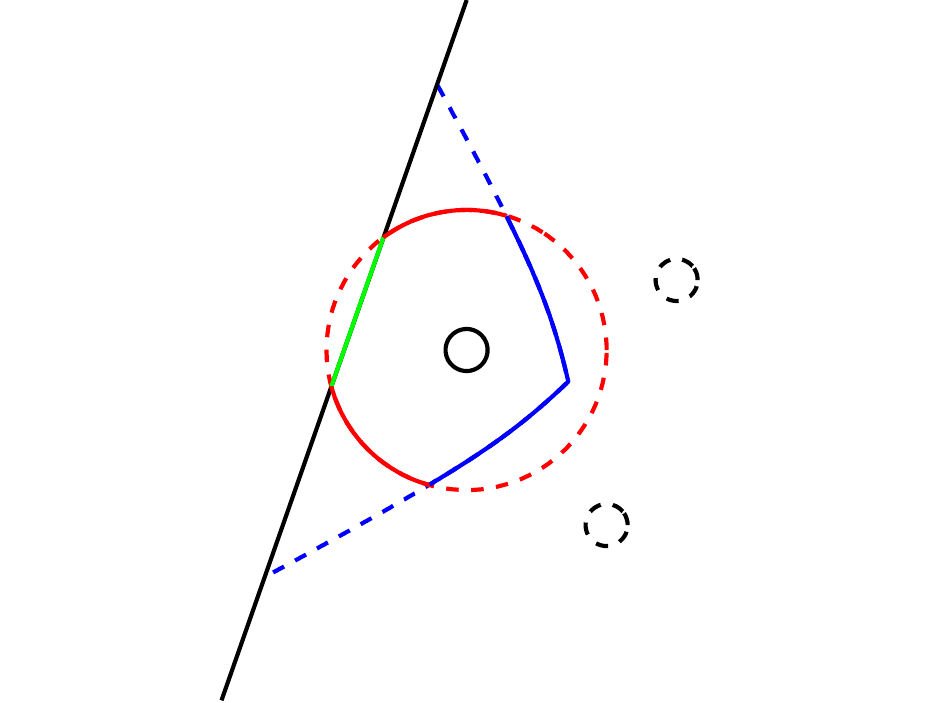}
		\caption{The decomposition of $\partial V_i^{gs}$ into disjoint sets (solid red, green and blue).}
		\label{fig:boundary_decomposition}
	\end{figure}

	Hence $\frac{\partial \mathcal{H}}{\partial x_i}$ can be written as
	\begin{eqnarray}
		\frac{\partial \mathcal{H}}{\partial x_i} &=&
		\int_{\partial V_i^{gs} \cap \partial \Omega}\tilde{\upsilon}_i^i\,\tilde{n}_i\,\phi\,\mathrm{d}x +
		\int_{\partial V_i^{gs} \cap \partial C_i^{gs}}\tilde{\upsilon}_i^i\,\tilde{n}_i\,\phi\,\mathrm{d}x + \nonumber \\
		&& \hspace{-1cm} \int_{\partial V_i^{gs} \cap \partial \mathcal{O}} \tilde{\upsilon}_i^i\,\tilde{n}_i\,\phi\,\mathrm{d}x +
		\sum_{j \in N_i^g}\int_{\partial V_j^{gs} \cap \partial \mathcal{O}} \tilde{\upsilon}_j^i\,\tilde{n}_j\,\phi\,\mathrm{d}x.
	\end{eqnarray}
	It is apparent that $\tilde{\upsilon}_i^i=0$ at $x\in \partial V_i^{gs}\cap\Omega$ since all $x\in\partial \Omega$ remain unaltered by infinitesimal motions of $x_i$, assuming no alteration of the environment.
	Considering the second integral, for any point $x \in \partial\tilde{V}_i^{gs}\cap\partial C_i^{gs}$ it holds that $\tilde{\upsilon}_i^i\left(x\right)=\mathbb{I}_2$, where $\mathbb{I}$ stands for the identity matrix, since they translate along the direction of motion of $x_i$ at the same rate.
	As far as the sets $\partial V_i^{gs} \cap \partial \mathcal{O}$ and $\partial V_j^{gs} \cap \partial \mathcal{O},~j\in N_i^g$ are concerned, they can be merged in pairs via utilization of the left and right hyperbolic branches, as introduced in GV--partitioning (\ref{GV_cell}), as follows
	\begin{eqnarray}
		\nonumber
		\int_{\partial V_i^{gs} \cap \partial \mathcal{O}} \tilde{\upsilon}_i^i\,\tilde{n}_i\,\phi\,\mathrm{d}x +
		\sum_{j \in N_i^g}\int_{\partial V_j^{gs} \cap \partial \mathcal{O}} \tilde{\upsilon}_j^i\,\tilde{n}_j\,\phi\,\mathrm{d}x= \\
		\sum_{j \in N_i^g}\left[\int_{\partial V_i^{gs} \cap \partial H_{ij}} \tilde{\upsilon}_i^i\,\tilde{n}_i\,\phi\,\mathrm{d}x + \int_{\partial V_j^{gs} \cap \partial H_{ji}} \tilde{\upsilon}_j^i\,\tilde{n}_j\,\phi\,\mathrm{d}x\right]
	\end{eqnarray}

	However, for any two Delaunay neighbors $i,j$ it can be observed that
	\begin{itemize}
		\item The hyperbolic branches $\partial H_{ij}$ and $\partial H_{ji}$ are symmetric with respect to the perpendicular bisector of $x_i,x_j$, and are governed by the same set of parametric equations (left and right branch).
		\item The vectors $n_j$ are mirrored images of the corresponding $n_i$ with respect to the perpendicular bisector of $x_i,x_j$.
	\end{itemize}

	Taking into account the above, $\frac{\partial\mathcal{H}}{\partial x_i}$ can be written as
	\begin{eqnarray}
		\nonumber
		\frac{\partial\mathcal{H}}{\partial x_i}&=&\int_{\partial V_i^{gs} \cap\partial C_i^{gs}}n_i\,\phi\,\mathrm{d}x + \\
		&& \hspace{-2cm}\sum_{j \in N_i^g}\left[\int_{\partial V_i^{gs} \cap \partial H_{ij}} \tilde{\upsilon}_i^i\,\tilde{n}_i\,\phi\,\mathrm{d}x + \int_{\partial V_j^{gs} \cap \partial H_{ji}} \tilde{\upsilon}_j^i\,\tilde{n}_j\,\phi\,\mathrm{d}x\right]
		\label{criterion_derivative}
	\end{eqnarray}

	The unit normal vectors $n_i,n_j$ and Jacobian matrices $\upsilon_i^i,\upsilon_j^i$ in the second part of (\ref{criterion_derivative}) can be evaluated via utilization of the parametric representations of the sets over which the integration takes place.
	These sets are parts of the left ($\partial H_{ij}$) and right ($\partial H_{ji}$) branch of the hyperbola that assigns the space among two arbitrary nodes $i,j$.
	The second term of the sum in (\ref{criterion_derivative}) requires knowledge of all the nodes in $\bigcup_{j \in N_i^g} N_j^g \subseteq I_n$ and so the control law is distributed.

	The decomposition of the set $\partial V_i^{gs} \cap \partial \mathcal{O}$ of node $i$ into mutually disjoint hyperbolic arcs $\partial V_i^{gs} \cap \partial H_{ij}$ and $\partial V_i^{gs} \cap \partial H_{ik}$ is shown in Figure \ref{fig:control_decomposition}.
	Figure \ref{fig:control_decomposition} also shows the other domains of integration used in (\ref{criterion_derivative}) $\partial V_i^{gs} \cap \partial C_{i}^{gs}$, $\partial V_j^{gs} \cap \partial H_{ji}$ and $\partial V_k^{gs} \cap \partial H_{ki}$.

	\begin{figure}[htb]
		\centering
		\includegraphics[width=0.4\textwidth]{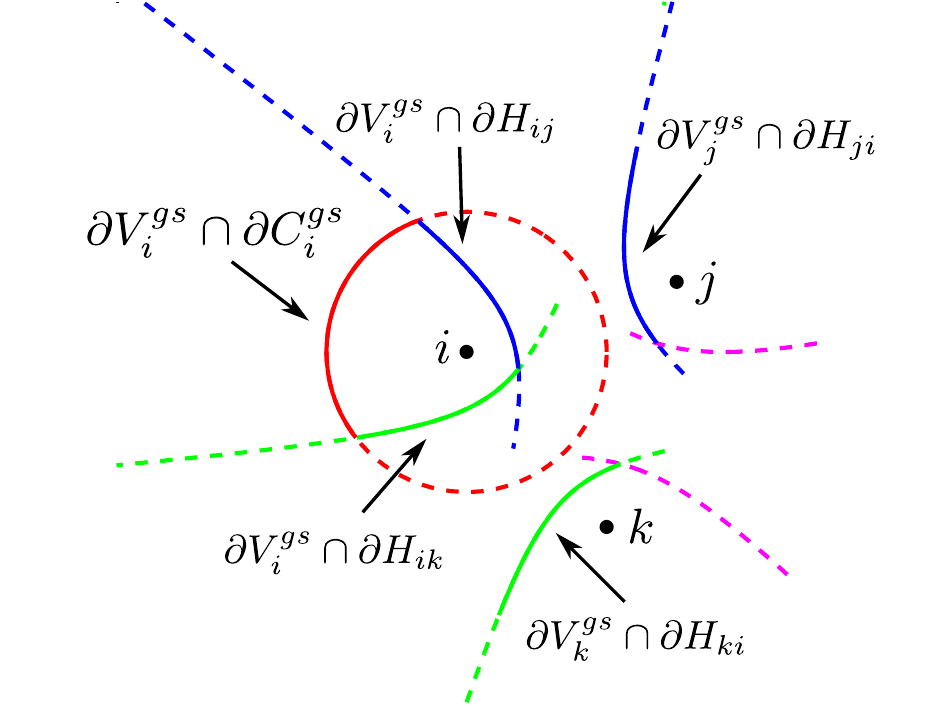}
		\caption{The domains of integration of the control law (\ref{control_law}).}
		\label{fig:control_decomposition}
	\end{figure}

	The proposed law (\ref{control_law}) leads to a gradient flow of $\mathcal{H}$ along the nodes trajectories, while $\mathcal{H}$ increases monotonically, since
	\[\frac{d\mathcal{H}}{dt}=\sum_{i\in I_n}\frac{\partial\mathcal{H}}{\partial x_i}\cdot \dot{x}_i=\sum_{i\in I_n}\left\|u_i\right\|^2\geq 0.\]
\end{proof}

\begin{remark}
	The resulting control law further distances nodes $i$ and $j$, since these move in a direction away from the bisector of $x_ix_j$.

	Consider the example of Figure \ref{fig:collision_avoidance} where the dotted hyperbolic branch of node $i$ has been moved so that its relative position with the node's center is the same in order to show the change in sensed area.
	It can easily be shown that the integral over $\partial V_i^{gs} \cap\partial C_i^{gs}$ of the control law will always result in a direction of movement for node $i$ away from the bisector of $x_ix_j$.
	It is also evident from the same figure that as node $i$ moves away from node $j$, the area of their cells $V_i^g, V_j^g$ also increases.
	As a result the covered area of both agents $V_i^{gs}$ and $V_j^{gs}$ increases and the integrals over $\partial V_i^{gs} \cap \partial H_{ij}$ and $\partial V_j^{gs} \cap \partial H_{ji}$ result in a direction of movement for agent $i$ away from the bisector, since the control law guarantees that each node's covered area will increase monotonically.

	\begin{figure}[htb]
		\centering
		\includegraphics[width=0.23\textwidth]{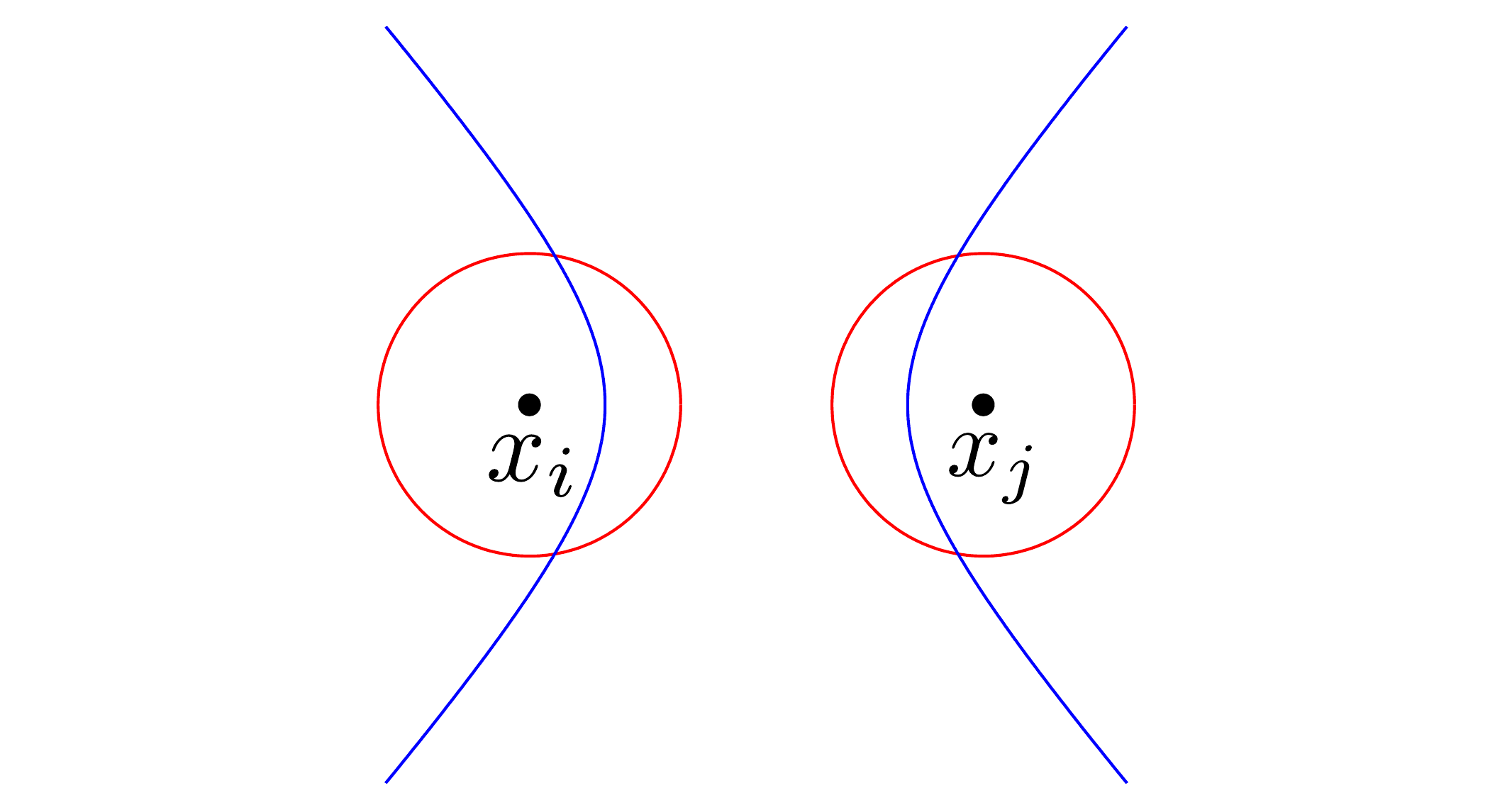}\hspace{0.01cm}
		\includegraphics[width=0.23\textwidth]{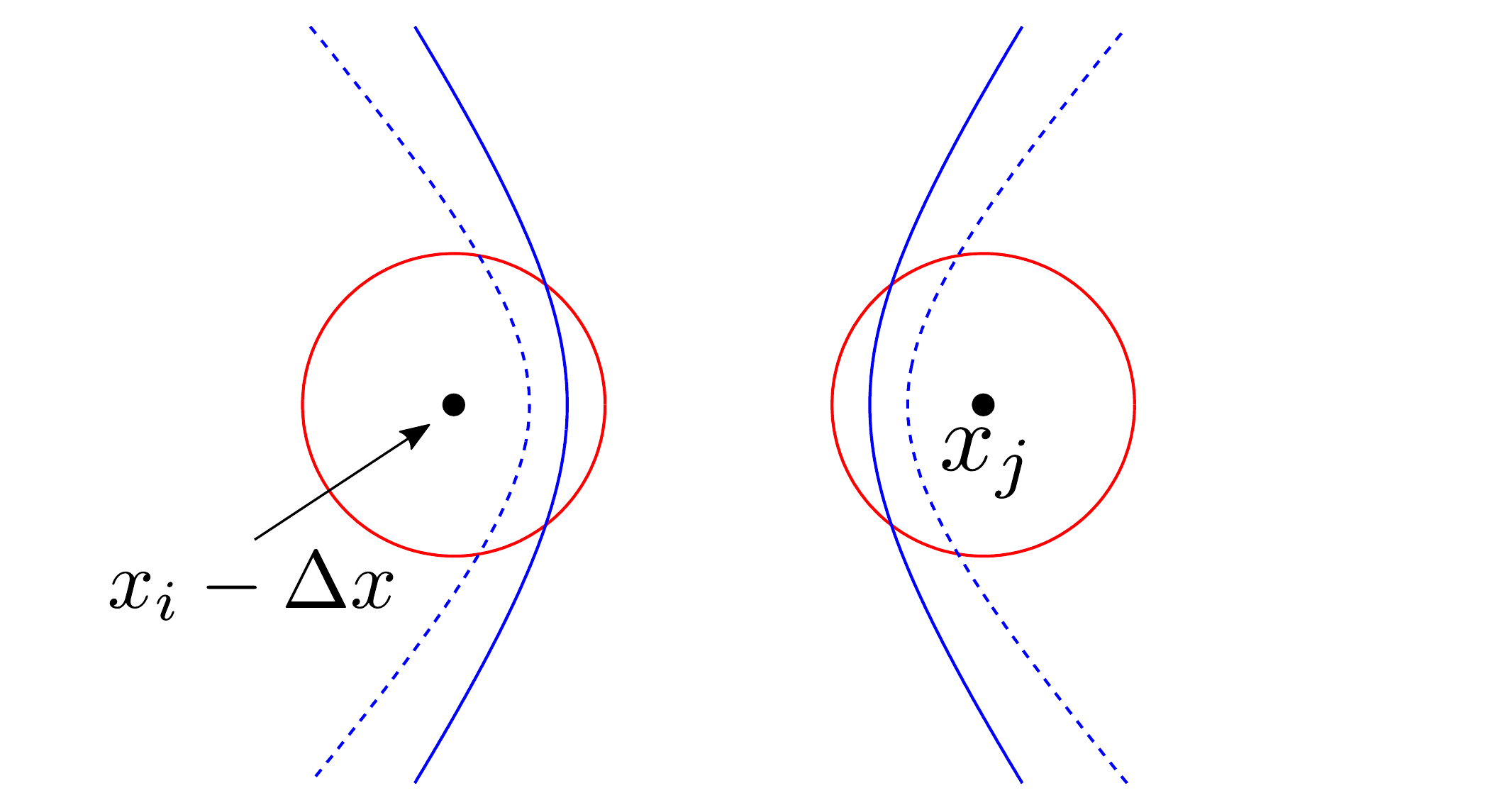}
		\caption{Sensed area change for a relative motion of the left node.}
		\label{fig:collision_avoidance}
	\end{figure}
\end{remark}

\begin{remark}
	The proposed control law can adapt to insertion, removal or immobilization of one or more of the network's nodes due to its distributed nature.
	It can also adapt to simultaneous changes to all agents of their common sensing radius $r^s$.
\end{remark}

\begin{conjecture}
	The computationally efficient suboptimal control law
	\begin{equation}
		u_i = \alpha_i \int_{\partial V_i^{gs} \cap\partial C_i^{gs}}n_i\,\phi\,\mathrm{d}x
		\label{control_suboptimal}
	\end{equation}
	leads the nodes to suboptimal trajectories while maintaining monotonous coverage increase and collision avoidance.
\end{conjecture}

As shown earlier, the direction of the vectors resulting from the integrals over $\partial V_i^{gs} \cap\partial C_i^{gs}$, $\partial V_i^{gs} \cap \partial H_{ij}$ and $\partial V_j^{gs} \cap \partial H_{ji}$ of (\ref{control_law}) is the same with respect to the bisector of $x_ix_j$.
Thus keeping only the first term of (\ref{control_law}) will also result in a movement of node $i$ away from the bisector.

\section{Simulation Studies}
\label{section:simulation_studies}

A network of $n=10$ agents with common positioning uncertainty circle radius $r^u = 0.05$ and common sensing radius $r^s = 0.3$ was used in the simulations.
The a-priori importance was $\phi(x) = 1, ~\forall x \in \Omega$ The region $\Omega$ under consideration is the same as in~\cite{Stergiopoulos_IETCTA10}.
The initial network configuration for all the case studies is shown in Figure \ref{fig:initial_network}.
For these simulation studies the suboptimal control law (\ref{control_suboptimal}) was used.

\begin{figure}[htb]
	\centering
	\includegraphics[width=0.3\textwidth]{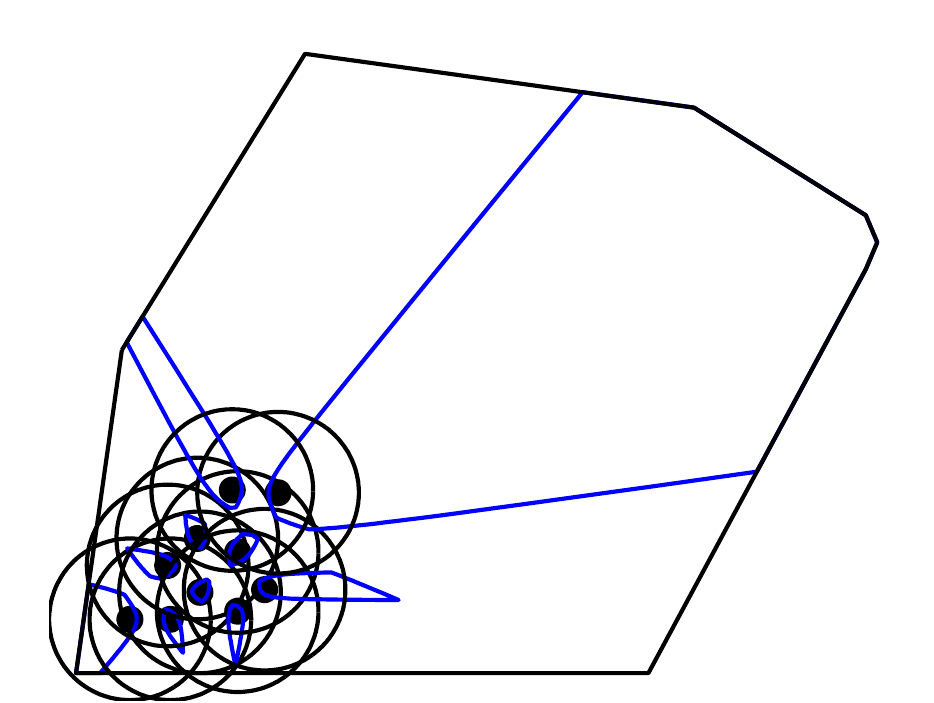}
	\caption{Initial network configuration for all case studies.}
	\label{fig:initial_network}
\end{figure}

\subsection{Case Study I}
The agents start concentrated in one corner of $\Omega$ and spread until they reach a local maximum of the objective (\ref{Hgc decomposition}).
We observe that the agents that move first are those on the "boundary" of the group, thus allowing the rest of the agents to move without collisions.
Some of the agents slide along the boundary of $\Omega$ resulting in the linear trajectories that can be seen in Figure \ref{fig:case1_network}.
This is expected behavior since the nodes move away from each other while at the same time increasing their respective covered area.
If a node moved in a way that resulted in it's sensing region to be partially outside $\Omega$, it's covered area would decrease and so the nodes' sensing regions slide along the boundary of $\Omega$.
In the final network configuration as seen in Figure \ref{fig:case1_network} all of the agent's sensing regions are completely contained within their GV cells and thus the maximum possible coverage is achieved.
This is not guaranteed to happen in every situation as the control law leads the network to a local maximum of the objective function.
In Figure \ref{fig:case12_area} we observe that the covered area increases monotonically as was expected.

\begin{figure}[htb]
	\centering
	\includegraphics[width=0.23\textwidth]{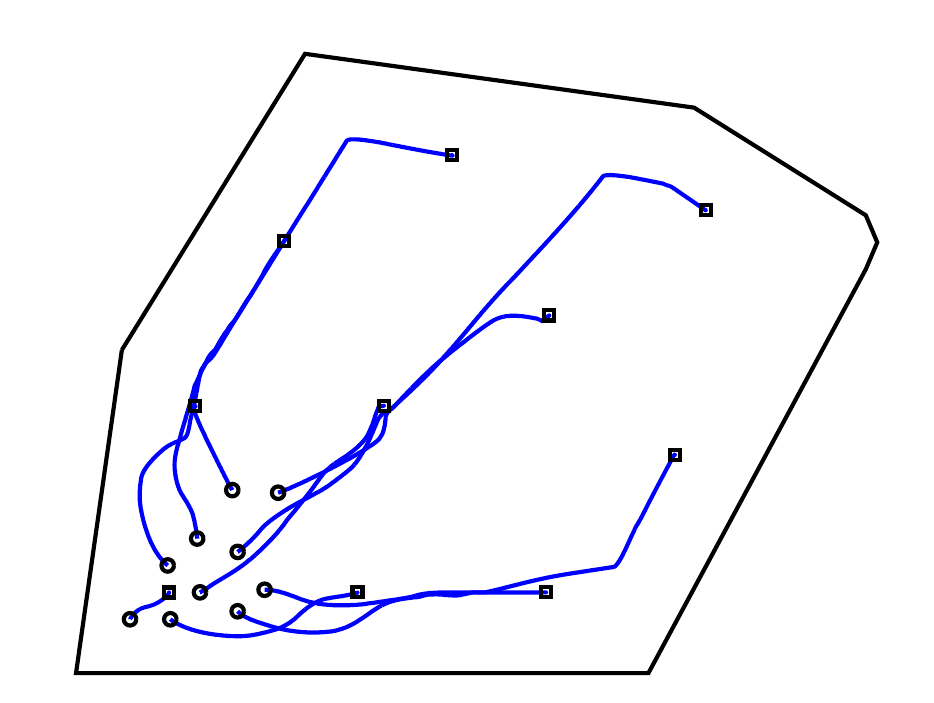}\hspace{0.01cm}
	\includegraphics[width=0.23\textwidth]{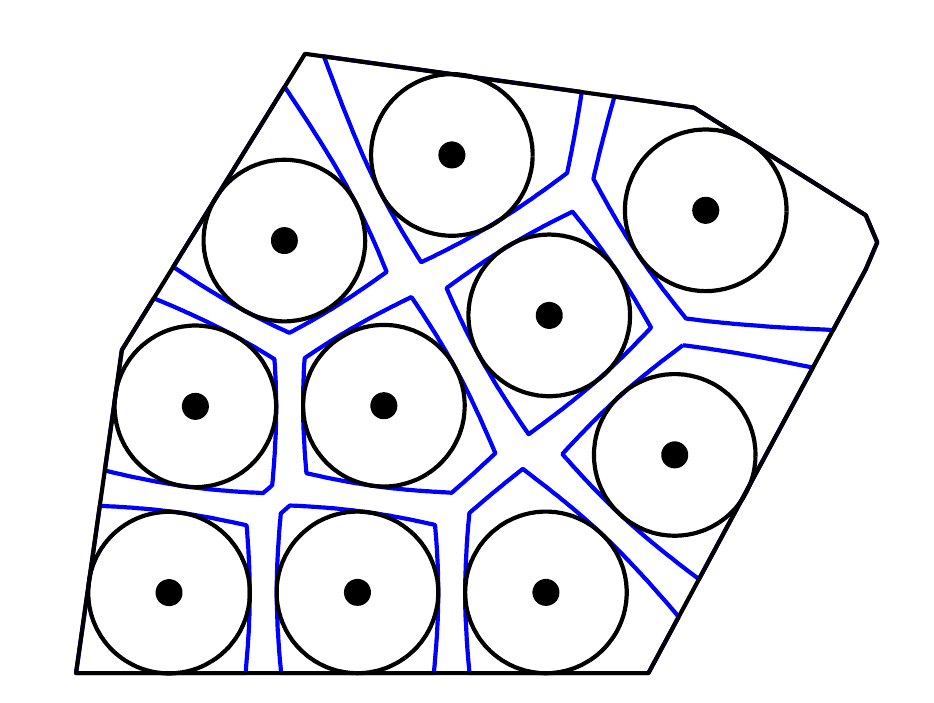}
	\caption{Case Study I [Left]: Agents' disk-center trajectories. [Right] Final network configuration.}
	\label{fig:case1_network}
\end{figure}

\subsection{Case Study II}
The initial network configuration is the same as in Case Study I.
At some point during the simulation before convergence, one of the agents stops moving while still having a functioning sensor, simulating a motor failure.
Despite this, the other agents achieve area coverage and avoid collisions between them and with the immobilized agent.
As it is expected, the network takes longer to converge since the agents have to move around the immobilized one.
Because of the immobile agent, the coverage of $\Omega$ is 98.4 \% of the maximum possible as seen in Figure \ref{fig:case12_area}.

A video of the two case studies can be found in:
\newline
\url{https://sotiris.papatheodorou.xyz/papers/2016_MED_PST/2016_MED_PST.mp4}

\begin{figure}[htb]
	\centering
	\includegraphics[width=0.23\textwidth]{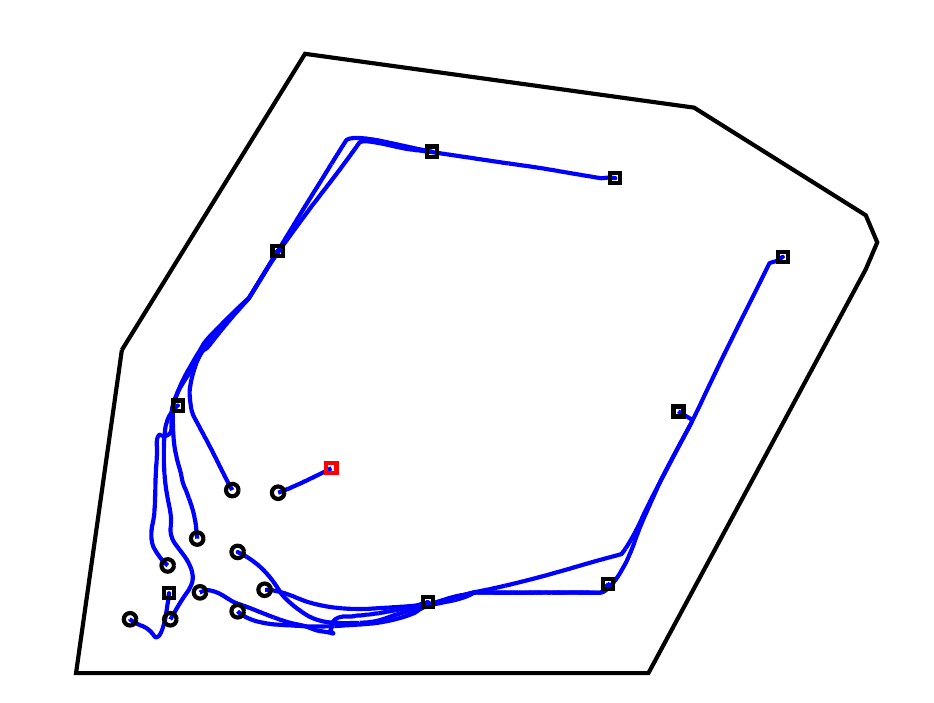}\hspace{0.01cm}
	\includegraphics[width=0.23\textwidth]{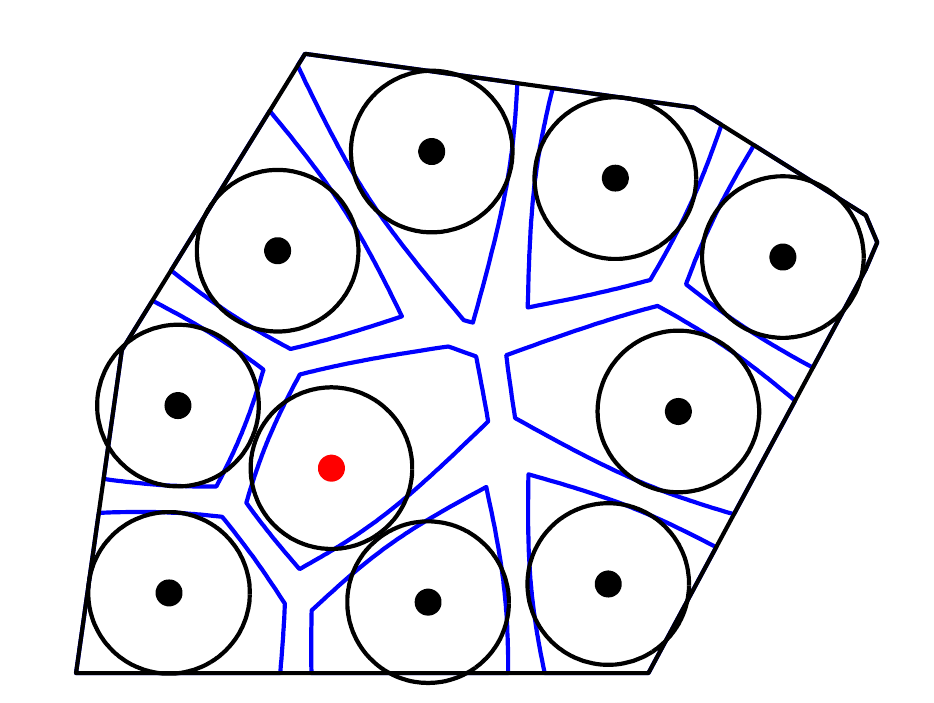}
	\caption{Case Study II The immobilized agent is shown in red. [Left]: Agents' disk-center trajectories. [Right] Final network configuration.}
	\label{fig:case2_network}
\end{figure}

\begin{figure}[htb]
	\centering
	\includegraphics[width=0.35\textwidth]{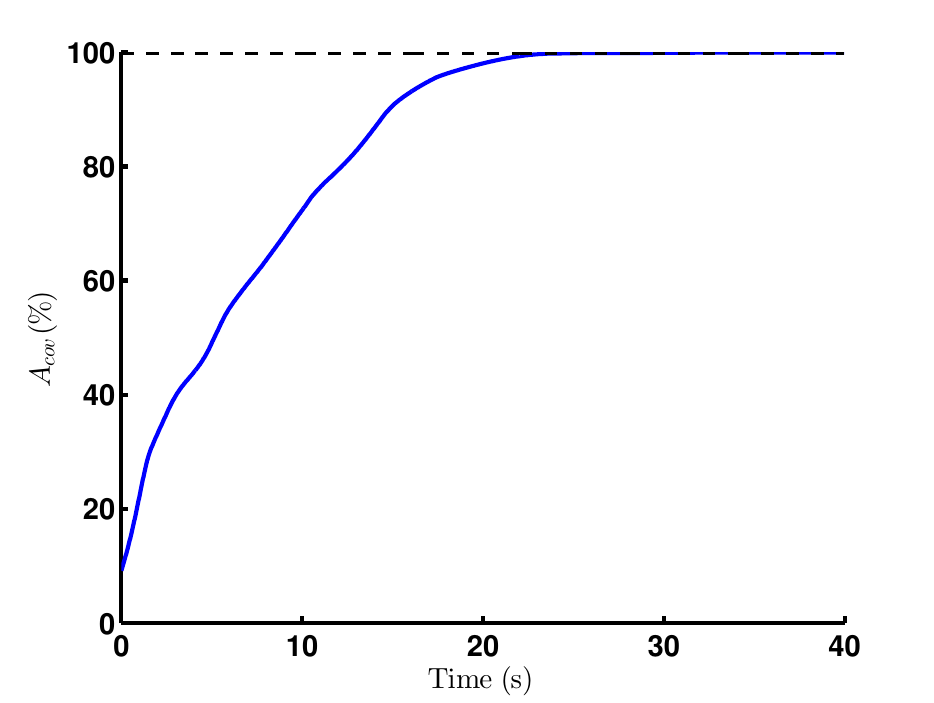}
	\includegraphics[width=0.35\textwidth]{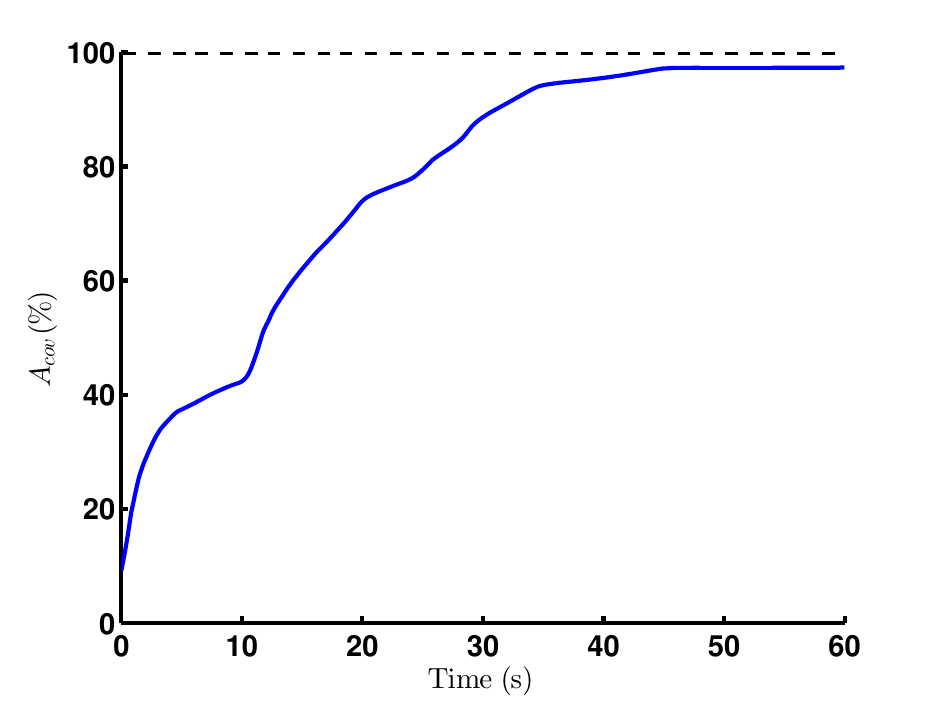}
	\caption{Percentage of the maximum possible covered covered area. $\bigcup_{i\in I_{n}}C_i^{gs}$ is shown in black and $\bigcup_{i\in I_{n}}V_i^{gs}$ in blue. [Top]: Case Study I. [Bottom] Case Study II.}
	\label{fig:case12_area}
\end{figure}

\section{Conclusions}
The area coverage for a homogeneous network of mobile robots with imprecise localization are examined in this article.
A gradient-ascent based scheme is used for area coverage in conjunction with a Guaranteed Voronoi tessellation of the area under surveillance.
This control scheme has an inherent collision avoidance property and no additional switching control laws are required.
Simulation studies are presented to show the efficiency and robustness of the proposed control method.

\bibliographystyle{IEEEtran}
\bibliography{references}
\end{document}